\def\draft{0}

\documentclass[11pt]{article}
\usepackage{amsfonts, amsmath, amssymb, url, graphics}
\usepackage[dvips]{graphicx}
\usepackage[latin1]{inputenc}
\usepackage[T1]{fontenc}
\usepackage{color}
\usepackage[dvips]{graphicx}
\usepackage{fullpage}

\newtheorem{theorem}{Theorem}[section]

\newtheorem{definition}[theorem]{Definition}
\newtheorem{lemma}[theorem]{Lemma}
\newtheorem{proposition}[theorem]{Proposition}

\newtheorem{remk}[theorem]{Remark}

\def\FullBox{\hbox{\vrule width 8pt height 8pt depth 0pt}}

\def\qed{\ifmmode\qquad\FullBox\else{\unskip\nobreak\hfil
\penalty50\hskip1em\null\nobreak\hfil\FullBox
\parfillskip=0pt\finalhyphendemerits=0\endgraf}\fi}

\newenvironment{proof}{\begin{trivlist} \item {\bf Proof:~~}}
  {\qed\end{trivlist}}

\def\qedsketch{\ifmmode\Box\else{\unskip\nobreak\hfil
\penalty50\hskip1em\null\nobreak\hfil$\Box$
\parfillskip=0pt\finalhyphendemerits=0\endgraf}\fi}

\newcommand{\etal}{{\it et~al.\ }}

\newcommand{\zo}{\{0,1\}}

\newcommand{\eps}{\varepsilon}

\ifnum\draft=1
\newcommand{\authnote}[2]{{ \bf [#1's Note: #2]}}
\else
\newcommand{\authnote}[2]{}
\fi

\newcommand{\COMMENT}[1]{}
\newcommand{\ket}[1]{|#1\rangle}

\newcommand{\ketbra}[2]{|#1\rangle\langle#2|}
\def\01{\{0,1\}}
\newcommand{\braket}[2]{\langle{#1}|{#2}\rangle} % inproduct, < | >
\def\01{\{0,1\}}

\newcommand{\triple}[3]{\langle{#1}|{#2}|{#3}\rangle}

\newcommand{\Tr}{\mbox{\rm Tr}}

\newcommand{\cadre}[1]
{
\begin{tabular}{|p{13.4cm}|}
\hline
#1 \\
\hline
\end{tabular}
}

\newcommand{\spa}[1]{\mathcal{#1}}

\newcommand{\pict}[1]{#1}

\title{Improved Loss-Tolerant Quantum Coin Flipping}
\author{André Chailloux}
\begin{document}

\maketitle
\begin{abstract}
In this paper, we present a loss-tolerant quantum strong coin flipping protocol with bias $\eps \approx 0.359$. This is an improvement over Berlin \etal's protocol~\cite{BBBG08} which achieves a bias of $0.4$. To achieve this, we extend Berlin \etal's protocol by adding an encryption step that hides some information to Bob until he confirms that he successfully measured. We also show using numerical analysis that a we cannot improve this bias by considering a $k$-fold repetition of Berlin \etal's protocol for $k > 2$.
\end{abstract}

\section{Introduction}

Coin flipping is a cryptographic primitive that enables two distrustful and far apart parties, Alice and Bob, to create a random bit that remains unbiased even if one of the players tries to force a specific outcome. It was first proposed by Blum ~\cite{Blu81} and has since found numerous applications in two-party secure computation.
In the classical world, coin flipping is possible under computational assumptions like the hardness of factoring or the discrete log problem. However, 
in the information theoretic setting, it is not hard to see that in any classical protocol, one of the players can always bias the coin to his or her desired outcome with probability 1.  

Quantum information has given us the opportunity to revisit information theoretic security in cryptography. The first breakthrough result was a protocol of Bennett and Brassard~\cite{BB84} that showed how to securely distribute a secret key between two players in the presence of an omnipotent eavesdropper. Thenceforth, a long series of work has focused on which other cryptographic primitives are possible with the help of quantum information. Unfortunately, the subsequent results were not positive. Mayers and Lo, Chau proved the impossibility of secure quantum bit commitment and oblivious transfer and consequently of any type of two-party secure computation~\cite{May97,LC97,DKSW07}. However, several weaker variants of these primitives have been shown to be possible~\cite{HK04,BCH+08}.

The case of coin flipping is one of the most intriguing ones. Even though the results of Mayers and of Lo and Chau exclude the possibility of perfect quantum coin flipping, it still remained open whether one can construct a quantum protocol where no player could bias the coin with probability 1. A few years later, Aharonov et al.~\cite{ATVY00} provided such a protocol where no dishonest player could bias the coin with probability higher than 0.9143. Then, Ambainis~\cite{Amb01} described an improved protocol whose cheating probability was at most $3/4$. Subsequently, a number of different protocols have been proposed~\cite{SR01,NS03,KN04} that achieved the same bound of $3/4$. Finally, it was shown in~\cite{CK09} how to achieve a strong coin flipping with cheating probability $\frac{1}{\sqrt{2}}$ using a weaker coin flipping primitive developed by Mochon~\cite{Moc07}.

The results mentioned earlier don't take into account practical issues such as losses, noise or other imperfections in the quantum apparatus used. In $2008$, Berlin \etal presented a loss-tolerant quantum coin flipping with bias $0.4$. In this protocol, honest players don't always succeed when they perform a measurement (the measurement sometimes abort) but when they do succeed, they always output the correct value. This is in contrast with noise tolerance where an honest player could perform a measure with a wrong outcome without knowing it. Recently, Aharon $\etal$~\cite{AMS10} created a loss-tolerant quantum coin flipping protocol with bias $\eps \approx 0.3975$. In another flavor, Barrett and Massar~\cite{BM04} showed how to do bit-string generation (a weaker notion of coin flipping) in the presence of noise.

In this paper, we continue the study of loss-tolerant quantum coin flipping protocol. We construct such a protocol with bias $\eps \approx 0.359$. To achieve this bias, we extend Berlin \etal's protocol by adding an encryption step that hides some information to Bob as long as he doesn't confirm that he successfully measured. Notice that we improve the bias of the protocol by adding only a classical layer on top of Berlin~\etal's protocol. Let us emphasize that in this paper we only look at
information theoretic security and we do not discuss computational security or security in
restricted models like the bounded-storage or noisy-storage model~\cite{DFSS08,WST08}.

It would be interesting whether to see whether such techniques can be used to deal with loss-tolerance in other practical models such as the bounded/noisy storage model. Moreover, finding a noise-tolerant quantum coin flipping with information theoretic security and small bias remains an interesting open question.

\section{Our work}

We continue \cite{BBBG08}'s work and try to create practical quantum coin flipping protocols. As their protocol, we ask Alice and Bob to send several copies of single qubit states. Moreover, we don't require honest players to have any quantum memories. On the other hand, we consider cheating players as being all powerful.

As explained in \cite{BBBG08}, one of the main difficulties in creating a bit-commitment based coin flipping lies in the states you send to Bob. The existence of a conclusive measurement between the states sent to Bob allows him to cheat perfectly even if the states are close. Berlin \etal's protocol is of the following form.
\begin{itemize}
\item Alice sends a state $\sigma$ to Bob.
\item Bob measures this state in some basis ${B}$ (possibly dependent on some of his  private coins). If Bob successfully measures then they continue the protocol. Otherwise, they start again
\end{itemize}
In this protocol, the state $\sigma$ is chosen very carefully such that a cheating Bob cannot take advantage of the fact, that he can reset the protocol.
This strongly limits the good choices for $\sigma$. To partially overcome this problem, we use the following high-level scheme
\begin{itemize}
\item Alice picks $r\in_R \zo$ and sends $E_r(\sigma)$ where $E_r$ is some quantum operation that hides some information about $\sigma$
\item Bob measures in some basis $B$. If Bob successfully measures then they continue the protocol. Otherwise, they start again
\item Alice reveals $r$ and then they continue the protocol
\end{itemize}
While doing this, one must be careful that an honest Bob will still be able to exploit the measurement of the encrypted state and that Alice cannot use this to cheat.

Applying this scheme on a two-fold parallel repetition of Berlin {\it{etal}}'s protocol, we show the following
\begin{theorem}
There is a loss-tolerant quantum coin flipping protocol with bias $\eps \approx 0.359$
\end{theorem}

Notice that without this encryption step extra step, the resulting scheme would not be loss-tolerant but the bias would remain the same.

\section{Preliminaries}
\subsection{Definitions}

The statistical distance over classical distributions is defined as
$\Delta(\{X_i\}_{i \in \zo^n},\{Y_i\}_{i \in \zo^n}) = \frac{1}{2}\sum_i |X_i - Y_i|$

Following~\cite{NC00}, we define the fidelity and the trace distance for any quantum states $\rho,\sigma$ as follows
\begin{eqnarray*}
D(\rho,\sigma) & = & \frac{1}{2}|\sigma - \rho| \quad \textrm{with } |A| = \sqrt{A^\dagger A} \\
F(\rho,\sigma) & = & tr(\sqrt{\rho^{1/2}\sigma\rho^{1/2}})
\end{eqnarray*}

Note that the fidelity is sometimes defined as $(tr(\sqrt{\rho^{1/2}\sigma\rho^{1/2}}))^2$.

For two quantum states $\rho,\sigma$ such that,  $\rho = \sum_{i} p_i \ketbra{i}{i}$ and $\sigma = \sum_i q_i \ketbra{i}{i}$ we have $D(\rho,\sigma) = \Delta(\{X_i\},\{Y_i\})$ and $F(\rho,\sigma) = \sum_{i} \sqrt{p_iq_i}$.

\begin{definition}
Let $E$ and $F$ any two ensembles of quantum states and let $\rho$ any quantum state. We define:
\begin{eqnarray*}
F(\rho,E) & = & \max_{\sigma \in E} F(\rho,\sigma) \\
F(E,F) & = & \max_{\sigma \in E, \sigma' \in F} F(\sigma,\sigma')
\end{eqnarray*}
\end{definition}
Finally, we define a (strong) coin flipping protocol
\begin{definition}
A coin flipping protocol with bias $\eps$ consists of instructions given to Alice and Bob and an outcome $x \in \{0,1,\bot\}$ ($\bot$ corresponds to aborting the protocol) such that
\begin{itemize}
\item If Alice and Bob are honest then $\Pr[x = 0] = \Pr[x = 1] = 1/2$
\item For any cheating Alice $P^*_A = \max\{\Pr[x = 0],\Pr[x = 1]\} \le 1/2 + \eps$
\item For any cheating Bob $P^*_B = \max\{\Pr[x = 0],\Pr[x = 1]\} \le 1/2 + \eps$
\end{itemize}
\end{definition}

For completeness, we also define weak coin-flipping protocols

\begin{definition}
A weak coin flipping protocol with bias $\eps$ consists of instructions given to Alice and Bob and an outcome $x \in \{0,1,\bot\}$ ($\bot$ corresponds to aborting the protocol) such that
\begin{itemize}
\item If Alice and Bob are honest then $\Pr[x = 0] = \Pr[x = 1] = 1/2$
\item For any cheating Alice $P^*_A = \Pr[x = 0] \le 1/2 + \eps$
\item For any cheating Bob $P^*_B = \Pr[x = 1] \le 1/2 + \eps$
\end{itemize}
Intuitively, $x=0$ corresponds to the fact that Alice wins and $x=1$ to the fact that Bob wins. Note that a cheating player can win with probability less than $1/2 + \eps$ but can lose with probability $1$.
\end{definition}

\subsection{Useful facts}
\begin{proposition}~\cite{NC00}\label{DistanceBound}
For any two states $\rho_0,\rho_1$ and any pure state $\ket{\phi}$, we have \[
D(\rho_0,\rho_1) \ge |\triple{\phi}{\rho_0}{\phi} - \triple{\phi}{\rho_1}{\phi}|
\]
\end{proposition}
\begin{proposition}~\cite{NC00}\label{ConvOfTraceDistance}
For any two states $\rho,\sigma$ such that $\rho = \sum_i p_i \ketbra{\phi_i}{\phi_i}$ and $\sigma = \sum_i q_i \ketbra{\phi_i}{\phi_i}$, we have
\[
D(\rho,\sigma) \le \Delta(\{X_i\},\{Y_i\}) \]
\end{proposition}
\begin{proposition}~\cite{FV99}\label{FidelityInequality1}
For any quantum states $\rho,\sigma$, we have
\[
1 - F(\rho,\sigma) \le D(\rho,\sigma) \le \sqrt{1 - F^2(\rho,\sigma)}
\]
\end{proposition}
\begin{proposition}~\cite{KN04}\label{FidelityInequality2}
For any quantum states $\rho,\sigma_0,\sigma_1$, we have
\[
F^2(\rho,\sigma_0) + F^2(\rho,\sigma_1) \le 1 + F(\sigma_0,\sigma_1)
\]
\end{proposition}
\begin{proposition}~\cite{Joz94}\label{AdditiveFidelity}
For any quantum states $\rho,\sigma_0,\sigma_1$, we have
\[
F^2(\rho,\sum_{i}p_i\sigma_i) \ge \sum_{i}p_i F^2(\rho,\sigma_i)
\]
\end{proposition}
\begin{proposition}~\cite{Hel67}
Suppose Alice has a bit $c \in_R \zo$ unknown to Bob. Alice sends a quantum state $\rho_c$ to Bob. We have
\[
\Pr[\mbox{Bob guesses } c] \le \frac{1}{2} + \frac{D(\rho_0,\rho_1)}{2}
\]
\end{proposition}

\section{The protocol}
\subsection{Quantum states used}
Consider the two orthonormal basis $\mathcal{B}^0(\lambda) = \{\ket{\phi^0_0(\lambda)},\ket{\phi^0_1(\lambda)}\}$ and $\mathcal{B}^1(\lambda) = \{\ket{\phi^1_0(\lambda)},\ket{\phi^1_1(\lambda)}\}$ for any $\lambda \in \mathbb{R}$ with:
\begin{eqnarray*}
\ket{\phi^0_0(\lambda)} & = & \sqrt{\lambda}\ket{0} + \sqrt{1-\lambda}\ket{1} \\
\ket{\phi^0_1(\lambda)} & = & \sqrt{1-\lambda}\ket{0} - \sqrt{\lambda}\ket{1}
\end{eqnarray*}
and 
\begin{eqnarray*}
\ket{\phi^1_0(\lambda)} & = & \sqrt{\lambda}\ket{0} - \sqrt{1-\lambda}\ket{1} \\
\ket{\phi^1_1(\lambda)} & = & \sqrt{1-\lambda}\ket{0} + \sqrt{\lambda}\ket{1} 
\end{eqnarray*}

$\ket{\phi_c^b}$ corresponds to the encoding of bit $c$ in basis $b$.

Finally, we define
\[\rho_c = \frac{1}{2}\sum_{i}\ketbra{\phi_c^i}{\phi_c^i} = \lambda\ketbra{c}{c} + (1-\lambda)\ketbra{1-c}{1-c}\]

\subsection{Berlin {\it{etal}}'s protocol}
\cadre{ \begin{center} Berlin {\it{etal}}'s protocol (parameter $\lambda$ omitted) \end{center}
\begin{enumerate}
\item Alice chooses at random $b \in_R \zo$ and $c \in_R \zo$ and sends $\ket{\phi_c^b}$ to Bob.
\item Bob chooses $b' \in_R \zo$ and measures the qubit he receives in basis $B_{b'}$. If his measurement fails, he announces it to Alice and they repeat the protocol from step 1. If the measurement succeeds continue.
\item Bob picks $c' \in_R \zo$ and sends $c'$ to Alice
\item Alice reveals $b,c$
\item If $b = b'$, Bob checks that what he measured corresponds to $\ket{\phi_c^b}$. If it doesn't match, he aborts.
\item The outcome of the protocol is $x = c \oplus c'$.
\end{enumerate}
} $ \ $ \\

This protocol is loss tolerant in the sense that a cheating Bob cannot gain advantage in the fact that he can restart the protocol when his measurement fails. This protocol has the following security parameters:
\begin{enumerate}
\item $P^*_A = \frac{3}{4} + \frac{\sqrt{\lambda(1-\lambda)}}{2}$
\item $P^*_B = \lambda$
\end{enumerate}
By taking $\lambda = 0.9$, we have $P^*_A = P^*_B = 0.9$ and  their protocol achieve a bias of $0.4$.
\subsection{Our protocol}
\cadre{ \begin{center} Our protocol \end{center}
\begin{enumerate}
\item Alice chooses at random $b_1,b_2 \in_R \zo$ ; $c \in_R \zo$ and $r_1,r_2 \in_R \zo$ sends two quantum registers $\ket{\phi_{c \oplus r_i}^{b_i}}$ for $i \in \{1,2\}$ to Bob.
\item Bob chooses $b'_1,b'_2 \in_R \zo$ and measures each register $i$ he receives in basis $B_{b'_i}$. If one of his measurements fails, he announces it to Alice and they repeat the protocol from step 1. If the measurement succeeds, Bob announces this fact to Alice and they continue.
\item Alice sends $r_1,r_2$ to Bob.
\item Bob picks $c' \in_R \zo$ and sends $c'$ to Alice
\item Alice reveals $b_1,b_2,c$
\item For each register $i$ for which $b_i = b'_i$, Bob checks that what he measured corresponds to $\ket{\phi_{c \oplus r_i}^{b_i}}$. If one of the measurements does not match, he aborts.
\item The outcome of the protocol is $x = c \oplus c'$.
\end{enumerate}
} $ \ $ \\

This protocol is closely related to a two-fold parallel repetition of Berlin {\it{etal}}'s protocol. Such a repetition would directly improve the bias if we did not require loss tolerance. We add an additionnal step in this protocol. Alice hides some information about the state she sends using $2$ private bits $r_1,r_2$ that she reveals as soon as Bob confirms that he measured successfully. As we will show, this makes the protocol loss-tolerant again.

\section{Security proofs}
If Alice and Bob are honest then Bob never aborts and $x = c \oplus c'$ is random. We now analyse separately cheating Alice and cheating Bob.

\subsection{Cheating Alice}
We consider a cheating Alice and an honest Bob. 

\subsubsection{General framework for checking Bob}
The way Bob checks is closely related to the following procedure
\begin{itemize}
\item Alice sends a state $\sigma$ in space $\spa{Y}$
\item At a later stage, Alice sends a bit $i$ to Bob in space $\spa{X}$
\item Bob checks that the first state Alice sends in $\spa{Y}$ is the state $\ket{\psi_i}$ for some state $\ket{\psi_i}$.
\end{itemize}
We want to show the following:
\begin{proposition}\label{MasterTheorem}
\[
\Pr[\textrm{ Alice passes Bob's test }] \le F^2(\sigma,L)
\]
where $L = \{\sum_{j}p_i\ketbra{\phi_j}{\phi_j} : \ \sum_{j}p_j = 1\}$ 
\end{proposition}

% Second Alternative proof

\begin{proof}
Let $\sigma$ the first state in $\spa{Y}$ sent by Alice and
let $\widetilde{\sigma}$ the state in $\spa{XY}$ after Alice reveals $i$. Since Bob immediately measures the register $\spa{X}$ in the computational basis, there is an state $\widetilde{\sigma}$ which will give the best cheating probability of the form
$
\widetilde{\sigma} = \sum_i p_i \ketbra{i}{i} \otimes \ketbra{\psi_i}{\psi_i} 
$ and 
\[ \Pr[\textrm{ Alice passes Bob's test }] = \sum_i |\ketbra{\psi_i}{\phi_i}|^2 \]
Similarly, if we fix $\widetilde{\sigma} = \ketbra{\Omega}{\Omega}$ where $\ket{\Omega} =  \sum_{i} \sqrt{p_i} \ket{i,\phi_i}$, we get that $ \Pr[\textrm{ Alice passes Bob's test }] = \sum_i |\ketbra{\psi_i}{\phi_i}|^2 $ This means that we can suppose w.log that after the last step, the state in $\spa{XY}$ is pure.

Let $\widetilde{\sigma} = \ketbra{\Omega}{\Omega}$ where $\ket{\Omega} =  \sum_{i} \sqrt{p_i} \ket{i,\phi_i}$. Let $K$ subspace of quantum pure states spanned by $\{\ket{i} \otimes \ket{\phi_i}\}$. Let $P_K = \sum_i \ketbra{i}{i} \otimes \ketbra{\phi_i}{\phi_i}$ the projection on subspace $K$. Bob's check is equivalent to projecting on the subspace $K$.
\begin{eqnarray*} 
\Pr[\textrm{ Alice passes Bob's test }] & = & tr(P_K\widetilde{\sigma}P_K) \\ 
& = & tr(P_K\ketbra{\Omega}{\Omega}P_K) = \max_{\ket{u} \in L} |\braket{\Omega}{u}|^2  \\  
& \le & \max_{\ket{u} \in K}F^2(\Tr_{\spa{X}}(\ketbra{\Omega}{\Omega}),\Tr_{\spa{X}}\ketbra{u}{u}) \\
& \le & \max_{\ket{u} \in K}F^2(\sigma,\Tr_{\spa{X}}\ketbra{u}{u}) \\ 
& \le & F^2(\sigma,L) \quad \textrm{since } \forall \ket{u} \in K, \ \Tr_{\spa{X}}\ketbra{u}{u} \in L
\end{eqnarray*} 
\end{proof}

\subsection{Proof of security for cheating Alice}~\label{CheatingAliceSec}
We consider a cheating Alice and an honest Bob. For the sake of the analysis, we can suppose that honest Bob doesnt' have losses when he measures (this does not help Alice). Our protocol says that Bob measures each register $i$ in a random basis $B_{b'_i}$ and performs a check if this basis corresponds to the basis $B_{b_i}$ in which Alice encoded $c$. Similarly, we could say that Bob performs this measurement at the very end (still picking $b'_i$ at random). In this case, we are in the framework of the previous subsection except that with some probability, Bob chooses the wrong basis and does not check anything.

Suppose Alice wants to reveal $c$ in our protocol. Let $\xi$ the state in $\spa{XY}$ she sends at state $1$. Let $\xi_X = \Tr_{\spa{Y}}\xi$ and $\xi_Y = \Tr_{\spa{X}}\xi$.
Let
$
L_c  =  \{\sum_{i \in \zo}p_i\ketbra{\phi^i_c}{\phi^i_c}\} $

We have the following cases:
\begin{itemize}
\item Bob flipped $b'_1 \neq b_1$ and $b'_2 \neq b_2$. Bob does nt check anything Alice successfully reveals $c$ with probability $1$.
\item Bob flipped $b'_1 = b_1$ and $b'_2 \neq b_2$. Bob checks the first register. From Proposition~\ref{MasterTheorem},  Alice successfully reveals $c$ with probability no greater than $F^2(\xi_X,L_c)$.
\item Bob flipped $b'_1 \neq b_1$ and $b'_2 = b_2$. Bob checks the second register. Similarly,  Alice successfully reveals $c$ with probability no greater than $F^2(\xi_Y,L_c)$.
\item Bob flipped $b'_1 = b_1$ and $b'_2 = b_2$. Bob checks both registers. In the same way,  Alice successfully reveals $c$ with probability no greater than $F^2(\xi,L_c^{\otimes 2})$.
\end{itemize}

This gives us
\[
\Pr[\textrm{ Alice successfully reveals } c] = \frac{1}{4}\left(1 + F^2(\xi_X,L_c) + F^2(\xi_Y,L_c) + F^2(\xi,L_c{\otimes 2})\right) \]

We will now need the following Lemma
\begin{lemma}\label{DistanceLemma}\[
F(L_0,L_1) \le 2\sqrt{\lambda(1-\lambda)} \]
\end{lemma}
\begin{proof}
Let $\rho_0 \in L_0$ and $\rho_1 \in L_1$ such that $F(\rho_0,\rho_1) = F(L_0,L_1)$. By definition of $L_0$, we have $\triple{0}{\rho_0}{0} = \lambda$ and $\triple{0}{\rho_1}{0} = 1 - \lambda$. This gives us $D(\rho_0,\rho_1) \ge 2\lambda - 1$. Using Proposition~\ref{FidelityInequality2}, we have 
\begin{eqnarray*}
F(\rho_0,\rho_1) & \le & \sqrt{1 - D^2(\rho_0,\rho_1)} \\
& \le & \sqrt{1 - 4\lambda^2 + 4\lambda - 1} \\
& \le & 2\sqrt{\lambda(1-\lambda)}
\end{eqnarray*}
\end{proof}
We can now prove our main statement
\begin{proposition}
\[ P^*_A \le \frac{1}{2} +\frac{1}{2}\left(\frac{1 + f(\lambda)}{2}\right)^2
\] 
where $f(\lambda) = 2\sqrt{\lambda(1-\lambda)}$
\end{proposition}
\begin{proof} 
We suppose w.log that Alice wants final outcome $x=0$.  This means that she has to reveal $c = c'$. Let $\xi$ the state sent by Alice and let $\xi_X = \Tr_{\spa{Y}}\xi$ and $\xi_X = \Tr_{\spa{Y}}\xi$. Since $c'$ is random, we have
\begin{eqnarray*}
P^*_A & = & \frac{1}{2}\sum_{c \in \zo} \Pr[\textrm{ Alice successfully reveals } c] \\
& \le & \frac{1}{2}\sum_{c \in \zo} \frac{1}{4}\left(1 + F^2(\xi_X,D_c) + F^2(\xi_Y,D_c) + F^2(\xi,DD_c)\right) \\
& \le & \frac{1}{8}\left(2 + 1 + F(D_0,D_1) + 1 + F(D_0,D_1) + 1 + F(DD_0,DD_1)\right) \ \ (Proposition~\ref{FidelityInequality1}) \\
& \le & \frac{1}{2} + \frac{1}{2}\left(\frac{1}{4} + \frac{1}{2}F(D_0,D_1) + \frac{1}{4}F^2(D_0,D_1)\right) \\
& \le &  \frac{1}{2} +\frac{1}{2}\left(\frac{1 + f(\lambda)}{2}\right)^2 \ (f(\lambda) \ge F(D_0,D_1) \textrm{ from Lemma}~\ref{DistanceLemma} )  
\end{eqnarray*}
\end{proof}

\subsection{Cheating Bob}
The main part here is to show the loss-tolerance of the protocol. This means that a cheating Bob cannot take advantage of the fact that he's allowed to reset the protocol in case one of his measurements failed.
\subsection{Cheat Sensitivity}
For a fixed $c$ and $r_1,r_2$, let $\xi^{r_1,r_2}_c$ sent by Alice. We have
\begin{eqnarray*}
\xi^{r_1,r_2}_c & = & \frac{1}{4}\sum_{ b_1,b_2 \in \zo} \ketbra{\phi_{c \oplus r_1}^{b_1} \phi_{c \oplus r_2}^{b_2}}{\phi_{c \oplus r_1}^{b_1} \phi_{c \oplus r_2}^{b_2}} \\
& = & \rho_{c \oplus r_1} \otimes \rho_{c \oplus r_2} \\
& = & \sum_{u,v \in \zo} p_{c \oplus r_1, c \oplus r_2}^{u,v}\ketbra{u,v}{u,v}
\end{eqnarray*}
where: if $x=y$ then $p_x^y = \lambda$ ; if $x\neq y$ then $p_x^y = 1 - \lambda$ and $ p_{c \oplus r_1, c \oplus r_2}^{u,v} = p_{c \oplus r_1}^u \cdot p_{c \oplus r_2}^v$.

When receiving $\xi$, Bob performs a quantum operation  
\[ A(\ket{u,v}) = \alpha_{u,v}\ket{\psi_{u,v}}\ket{0}_{\spa{O}} + \beta_{u,v}\ket{\omega_{u,v}}\ket{1}_{\spa{O}}\] where $\spa{O}$ is the space that Bob measures to determine whether he should announce that he succeeded the measurement or not. The outcome 0 in space $\spa{O}$ corresponds to the outcome where the protocol continues. In a way, the cheating Bob postselects on the outcome being $0$ since if he obtains $1$, he decides to start the protocol again. Once Bob successfully measured and after Alice sends $r_1,r_2$, Bob has the following state depending on the operation $A$ he performed averaging on  $r_1,r_2$.
\[
\xi_c^{A}  =  \frac{1}{S} \sum_{\substack{r_1,r_2 \in \zo \\ u,v \in \zo}}  p_{c \oplus r_1, c \oplus r_2}^{u,v} \Gamma_{u,v} \ketbra{r_1,r_2, \psi_{u,v}}{r_1,r_2, \psi_{u,v}} \]
where 
\begin{itemize}
\item The $\Gamma_{u,v}$'s are arbitrary real numbers. These numbers depend on the $\alpha_{u,v}$'s. We assume that Bob can choose any value for these numbers.
\item The $\ket{\psi_{u,v}}$'s are not necessarily orthogonal. 
\item $S$ is a normalization factor. 
\end{itemize}

\begin{proposition}
$\forall A$, $D(\xi_0^A,\xi_1^A) \le D(\xi_0,\xi_1)$ where $\xi_c = \rho^{\otimes 2}_c$. 
\end{proposition}
\begin{proof}
Let's fix $A$. We have
\[
D(\xi_0^A,\xi_1^A) = \\
\frac{1}{S}D(\sum_{\substack{r_1,r_2 \in \zo \\ u,v \in \zo}}  p_{r_1,r_2}^{u,v} \Gamma_{u,v} \ketbra{r_1,r_2, \psi_{u,v}}{r_1,r_2, \psi_{u,v}},\sum_{\substack{r_1,r_2 \in \zo \\ u,v \in \zo}}  p_{1 \oplus {r_1}, 1 \oplus {r_2}}^{u,v} \Gamma_{u,v} \ketbra{r_1,r_2, \psi_{u,v}}{r_1,r_2, \psi_{u,v}}) \]
from convexity of the statistical distance (Proposition~$\ref{ConvOfTraceDistance}$) , we have
\begin{eqnarray*}
D(\xi_0^A,\xi_1^A) & \le & \frac{1}{S}\Delta\left(\{ p_{r_1,r_2}^{u,v} \Gamma_{u,v}\}_{\substack{r_1,r_2 \in \zo \\ u,v \in \zo}},\{p_{1 \oplus {r_1}, 1 \oplus {r_2}}^{u,v} \Gamma_{u,v}\}_{\substack{r_1,r_2 \in \zo \\ u,v \in \zo}}\right) \\
& \le & \frac{1}{2S}\sum_{\substack{r_1,r_2 \in \zo \\ u,v \in \zo}}|p_{r_1,r_2}^{u,v} \Gamma_{u,v} - p_{1 \oplus {r_1}, 1 \oplus {r_2}}^{u,v} \Gamma_{u,v}| \\
& \le & \frac{1}{2S}\sum_{u,v}\Gamma_{u,v} \sum_{r_1,r_2}|p_{r_1,r_2}^{u,v} -  p_{1 \oplus {r_1}, 1 \oplus {r_2}}^{u,v}| 
\end{eqnarray*}
To calculate this sum, if $(r_1,r_2) = (u,v)$ then $p_{r_1,r_2}^{u,v} = \lambda^2$ and $p_{1 \oplus {r_1}, 1 \oplus {r_2}}^{u,v} = (1 - \lambda)^2$. If $(r_1,r_2) = (\overline{u},\overline{v})$ then $p_{r_1,r_2}^{u,v} = (1 - \lambda)^2$ and $p_{1 \oplus {r_1}, 1 \oplus {r_2}}^{u,v} = \lambda^2$. In the other cases, $p_{r_1,r_2}^{u,v} = p_{1 \oplus {r_1}, 1 \oplus {r_2}}^{u,v}$. 
This gives us
\begin{eqnarray*}
D(\xi_0^A,\xi_1^A) & \le & \frac{1}{2S}\sum_{u,v}2 \Gamma_{u,v}\left(\lambda^2 - (1-\lambda)^2\right) \\ 
& \le & 2\lambda - 1 
\end{eqnarray*}
Since, $\xi_c = \lambda^2 \ketbra{cc}{cc} + \lambda(1 -\lambda) (\ketbra{01}{01} + \ketbra{10}{10}) + (1-\lambda)^2\ketbra{\overline{c} \ \overline{c}}{\overline{c} \; \overline{c}}$, we have $
D(\xi_0,\xi_1) = (\lambda^2 - (1-\lambda)^2) = 2\lambda - 1$, which allows us to conclude.
\end{proof}

We can now prove our main Claim
\begin{proposition}
$P^*_B \le \lambda$
\end{proposition}
\begin{proof}
Suppose w.log that Bob wants outcome $x = 0$. He wants to pick $c' = c$. Before picking $c'$, he has the state $\xi_c^A$. We have
\begin{eqnarray*}
P^*_B & = & \Pr[\textrm{ Bob guesses } c] \\
& = & \frac{1}{2} + \frac{D(\xi_0^A,\xi_1^A)}{2} \\
& \le & \lambda
\end{eqnarray*}
\end{proof}

\begin{theorem}
There is a loss-tolerant quantum coin flipping protocol with bias $\eps \approx 0.359$
\end{theorem}
\begin{proof}
We just need to find $\lambda$ that minimizes $\max(P^*_A,P^*_B)$. The maximum is achieved for $\lambda \approx 0.859$ which gives $P^*_A = P^*_B \approx 0.859$ which gives a bias $\eps \approx 0.359$.
\end{proof}

\section{Further discussion}
\paragraph{Optimality of the bias} The bias that we show here is actually not optimal for the protocol. The reason is the following: in the analysis of cheating Alice (Section~\ref{CheatingAliceSec}), we consider the cheating probability for Alice depending on whether Bob checks the first bit, the second bit or both bits. For each of these cases, we upper bound Alice's cheating probability. But it appears that the cheating probabilities for each of these cases is different and that Alice cannot cheat optimally for all these cases at the same time. This slightly decreases Alice's cheating probability. We can numerically calculate calculate in this case that for $\lambda \approx 0.858$, we have $P^*_A = P^*_B \approx 0.858$. This gives a bias of $\eps \approx 0.858$ which is a slight improvement over what is shown.

\paragraph{Multiple repetition} Our protocol consists of a two-fold repetition of Berlin $\etal$'s protocol. What happens if we consider a $k$-fold repetition? Even if it is difficult to calculate the exact cheating probabilities of Alice and Bob in the case of multiple repetitions, these probabilities can be easily upper and lower bounded. We use the following bounds. Let $P^*_A(k,\lambda)$ the cheating probability for Alice (resp. Bob) with a $k$-fold repetition of Berlin $\etal$'s protocol with parameter $\lambda$. Let $P(k) = \min_{\lambda}(\max\{P^*_A(k,\lambda),P^*_B(k,\lambda)$. $P(k)$ corresponds to the best cheating probability when consider a $k$-fold repetition of the protocol. We need to lower bound $P^*_A(k,\lambda)$. We have
\[
P^*_A(k,\lambda) \le f(k,\lambda) = \frac{1}{2} + \frac{1}{2}\left(\frac{1}{2} + \sqrt{\lambda(1-\lambda)}\right)^k
\]

This is a generalization of the upper bound we use to show that $\eps \approx 0.359$. Intuitively, this corresponds to the case where Alice knows if Bob measured in the correct basis' or not. When we consider Alice's cheating strategies where she uses separate (non entangled) strategies for each of the $k$ repetitions, we have the following lower bound.
\[
P^*_A(k,\lambda) \ge g(k,\lambda) = (\frac{3}{4} + \frac{\sqrt{\lambda(1-\lambda)}}{2})^k \]
On the other hand, it possible to calculate exactly Bob's cheating probability since
\[
P^*_B(k,\lambda) = 1/2 + D(\rho_0^{\otimes k},\rho_1^{\otimes k})/2
\]
Using these bounds, we get the following diagram for cheating probabilities of Alice and Bob which shows that the optimal value is achieved using a $2$-fold repetition of the protocol. The $x$-axis corresponds to the number of repetition $k$. The $y$-axis corresponds to the minimal cheating probability $P(k)$ when using lower/upper bounds for $P^*_A$. \\ 
\begin{center}
\pict{\includegraphics[width = 7cm]{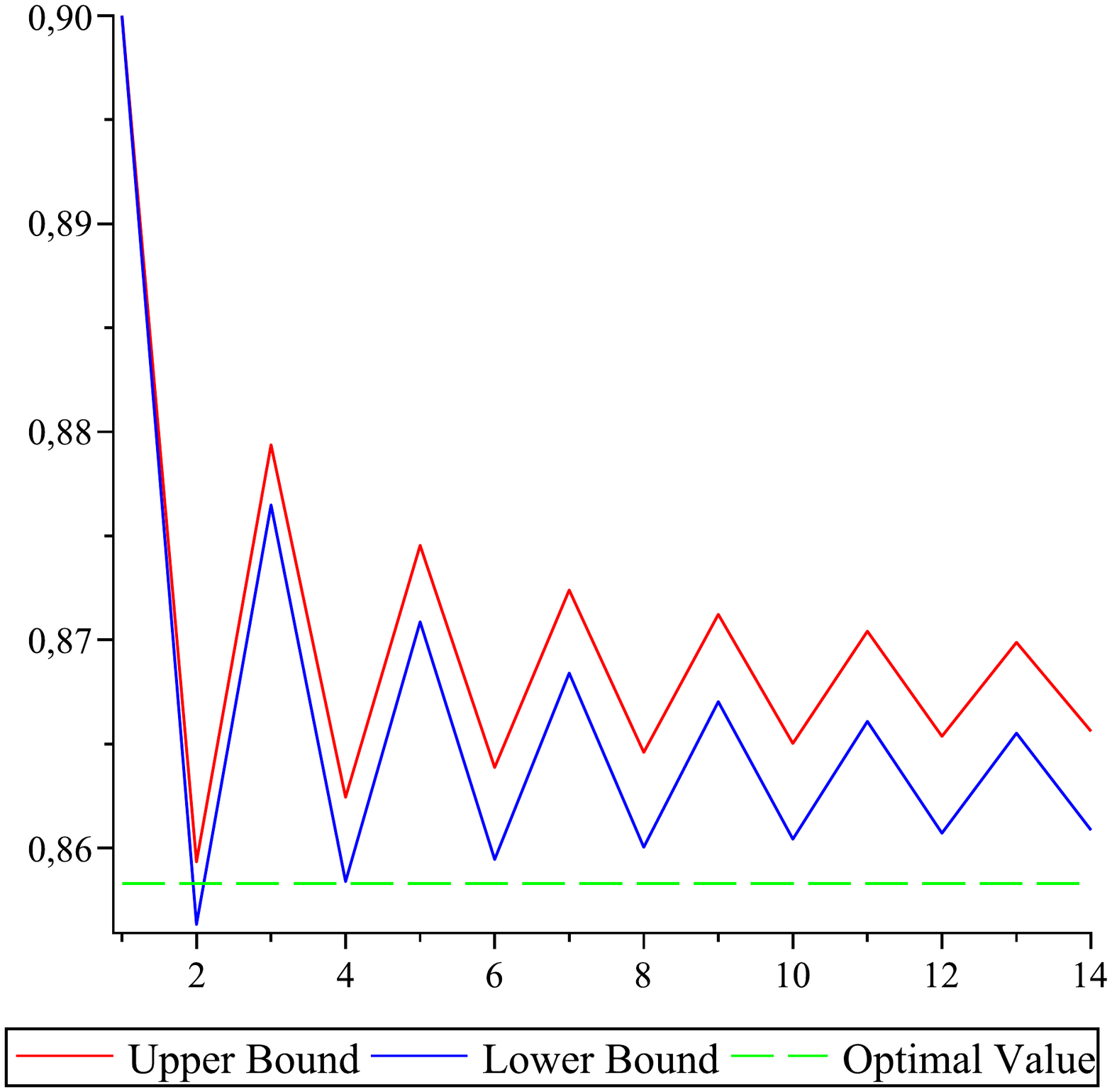}} \end{center}

\section{Conclusion and open questions}
In this work, we presented a loss-tolerant quantum coin flipping protocol with bias $\eps \approx 0.359$. To do this, we presented a general method to disallow cheating Bob to take advantage of the fact that he can reset the protocol when one of his measurement fails.
It would be interesting to see whether such techniques can be used for other protocols which have information theoretic security or not. Moreover, what is the best bias that can be achieved for such loss-tolerant protocols and can such protocols also be noise-tolerant?

\bibliography{paper}
\bibliographystyle{alpha}

\end{document}